\RequirePackage{fix-cm}
\documentclass[smallextended]{svjour3}       % onecolumn (second format)
\smartqed  % flush right qed marks, e.g. at end of proof
\usepackage{graphicx}
\begin{document}

\title{On asymptotic effects of boundary
perturbations in exponentially shaped Josephson junctions%\thanks{Grants or other notes
}

\titlerunning{On asymptotic effects of boundary perturbations in ESJJ}        % if too long for running head

\author{Monica De Angelis \and Pasquale Renno }


\institute{M. De Angelis  \at
              Univ. of Naples  "Federico II", Dip. Mat. Appl. "R.Caccioppoli" Via Claudio n.21, 80125, Naples, Italy. \\
               \email{modeange@unina.it} \and P. Renno \at
              Univ. of Naples  "Federico II", Dip. Mat. Appl. "R.Caccioppoli" Via Claudio n.21, 80125, Naples, Italy. \\
               \email{renno@unina.it}          
}

\date{Received: date / Accepted: date}

\maketitle

\begin{abstract}
A parabolic integro differential operator  $\mathcal L,$   	suitable to describe many   phenomena in various physical fields, is considered. By means of equivalence between   $\mathcal L$  and the third order  equation describing  the evolution inside an exponentially shaped Josephson junction (ESJJ),  an  asymptotic analysis  for (ESJJ) is achieved, explicitly evaluating,    boundary contributions related to  the Dirichlet problem.

\end{abstract}

\section{Introduction}
\label{intro}

 Many  equivalences  among  nonlinear  operators and p.d.e. systems exist and an extensive bibliography is given. (see, f.i. \cite{rb,acscott,acscott02}). Here,  the  semilinear equation, which characterizes exponentially shaped Josephson junctions  in superconductivity (ESJJ) (\cite{ddf,mda13,mda10,df13,df213} and references therein), is considered  and  equivalence  with  the following  parabolic integro differential equation:

\begin{equation}   \label{11}
%\left \{
  % \begin{array}{lll}
  \mathcal L u \equiv \,\, u_t -  \varepsilon  u_{xx} + au +b \int^t_0  e^{- \beta (t-\tau)}\, u(x,\tau) \, d\tau \,=\, f(x,t,u) \,
   \end{equation}  
  
  \noindent  is proved. In this way,   a priori estimates for (ESJJ) are obtained  and, by means of a well known  
theorem  on  convolutions behavior, 
  asymptotic effects of boundary perturbations are achieved for the solution of initial boundary value problem with Dirichlet conditions.

   Operator  $ {\cal L}  $  defined in  (\ref{11}) can  describe  many linear and non linear physical phenomena and there is plenty  of bibliography.\cite{bcf,r,dr1,mps,fr,l,dmr,s}. In particular, when $ F = F(x,t,u)$, some non linear phenomena involve equation (1) both in superconductivity and biology and  Dirichlet  conditions in superconductivity  refer to the phase boundary specifications \cite{mda10,df13,df213,de}, while in
excitable systems occur when
 %the behavior of a single dendrite has
%to be determined and the voltage level is fixed or when
 the pulse propagation in
 heart cells is studied\cite{ks}. Besides,  Dirichlet problem is also
considered  for stability analysis and asymptotic behavior of reaction-diffusion
systems solutions, \cite{t,fdd}, or in hyperbolic diffusion \cite{g}.

 Previous analyses related to equation (\ref{11}) have been developed in  \cite{dm13,de13,dr13,dr8} assuming  $ \varepsilon, a, b, \beta  $  as  positive constants. The fundamental solution $ K $ has been determined and many of  its properties have been proved. Moreover, various boundary value problems have been considered as well as non linear integral equations have been determined, whose Green functions have numerous typical properties of the diffusion equation.  Besides, the  kernel  $ e^{- \beta (t-\tau)}\, u(x,\tau) $  in (\ref{11}) can be modified   as physical situations demand, and this particular choice has been made taking into consideration superconductive and biological models.
As a matter of fact, it is possible to prove that  (\ref{11}) can  characterize  reaction diffusion models, like the FitzHugh-Nagumo system, suitable to  model the propagation  of nerve impulses. Furthermore, the perturbed sine Gordon equation (PSGE)  can be deduced and, as proved further,  the evolution in a  Josephson junction  with nonuniform width can be characterized, too.

   In superconductivity it is well known that the Josephson effect   is modeled by the (PSGE) given by (see, f.i. \cite{bp}):
\begin{equation}  \label{12}
\varepsilon u_{xxt}\, - \, u_{tt} \, +\, u_{xx}-\, \alpha u_t = \,  \, \sin u \,- \gamma   
\end{equation} 

\noindent where $u$ represents the difference between the phase of the wave functions related to the two superconductors of the junction, $ \gamma >0$ is a forcing term that is proportional to a bias current, the $\alpha$-term accounts for dissipation due to  normal electrons crossing the junction, and the 
$\varepsilon$-term accounts for dissipation  caused by normal electrons flowing parallel to the junction.

Besides, when the case of an exponentially shaped Josephson junction (ESJJ) 
 is considered, denoting by  $\, \lambda $  a   positive constant,  the evolution of the phase inside the junction is described by the third order equation:

 \begin{equation}            \label{13}
  \varepsilon u_{xxt}+u_{xx}-u_{tt} - \varepsilon\lambda u_{xt} -\lambda  u_x- \alpha u_{t}=\sin
  u-\gamma
 \end{equation}
 \noindent  where terms  $ \,\lambda \, u_x $ and $\,\lambda \, \varepsilon \,u_{xt} \,$ stand for the current as a consequence of the tapering and, more specifically, $ \lambda u _{x} $ represents a geometrical force leading the fluxons from the wide edge to the narrow edge. \cite{bcs96,bcs00,cmc02}. Exponentially shaped Josephson junctions provide several advantages compared to rectangular ones. Among the others, we can mention the possibility of getting  a voltage which is not chaotic anymore, but rather periodic. This allows to exclude some among the possible causes of large spectral width and to avoid the problem of trapped flux.   \cite{bss,j05,ssb04,j005}. 
 
\section{Statement of the problem}

\label{sec:2}

If   $\, T\, $  is   an arbitrary positive constant and   

\[
\,   \Omega_T \, \equiv \{\,(x,t) : \, 0\,\leq \,x \,\leq L \,\,;  \ 0 < t \leq T \,\}, \]

  \noindent let us consider the following initial boundary value problem with  Dirichlet boundary conditions:

\begin{equation}   \label{21}
\left \{
   \begin{array}{lll}
\,u_t -  \varepsilon  u_{xx} + au +b \int^t_0  e^{- \beta (t-\tau)}\, u(x,\tau) \, d\tau \, =  F(x,t,u) \, & (x,t) \in \Omega_T \,  \\    
\\  \,u (x,0)\, = u_0(x)\, \,\,\, &
x\, \in [0,L], 
\\
\\
  \, u(0,t)\,=\,g_1(t)  \qquad u(L,t)\,=\,g_2(t) & 0<t\leq T.
   \end{array}
  \right.
\end{equation}

Denoting by   $ \, K(x,t) \, $ the fundamental solution of the linear  operator defined by (\ref{11}), and considering   $ \varepsilon, a, b, \beta  $  as  positive constants, it results \cite {dr8}:

\begin{equation}     \label {22}
%\begin{split}
K(r,t)=  \frac{1}{2 \sqrt{\pi  \varepsilon } }\biggl[ \frac{ e^{- \frac{r^2 }{4 t}-at}}{\sqrt t}
 -b \int^t_0  \frac{e^{- \frac{r^2}{4 y}\,- a y}}{\sqrt{t-y}}   e^{-\beta ( t -y)}  J_1 (2 \sqrt{by(t-y)\,})dy \biggr]
%\end{split}
\end{equation}

\noindent where $\, r\,= |x| \, / \sqrt \varepsilon \, \, $ and   $ J_n (z) \,$    denotes the Bessel function of first kind and order $\, n.\,$ Moreover, the following theorem holds:

\begin{theorem}

For all $t>0$, the Laplace transform of  $\,K (r,t)\, \,$  with respect to $\,t\,$ converges absolutely in the half-plane $ \Re e  \,s > \,max(\,-\,a ,\,-\beta\,)\,$ and it results:

\begin{equation}      \label{23}
\,\hat K\,(r,s)  =\,\int_ 0^\infty e^{-st} \,\, K\,(r,t) \,\,dt \,\,=  \, \frac{e^{- \,r\,\sigma}}{2 \, \sqrt\varepsilon \,\sigma \,  }
\end{equation}

 \noindent with $\,\sigma^2 \ \,=\, s\, +\, a \, + \, \frac{b}{s+\beta}.$ 
\end{theorem}

Now, let us consider  the following  Laplace transforms with respect to $\,t\,$:

\[
\hat u (x,s) \, = \int_ 0^\infty \, e^{-st} \, u(x,t) \,dt \,\,, \,\,\,  \,\hat F (x,s)   \, = \int_ 0^\infty \,\, e^{-st} \,\, F\,[x,t,u (x,t)\,] \,dt \,,\
\]

 \noindent  and let $\hat g_1(s ), \,\,\, \hat g_2(s)\,\, $  be the  ${ L} $  transforms of the  data $ g_i(t ) \,\,(i=1,2).\, $

\noindent  Then the Laplace transform of the problem (\ref{21})  is formally given by:

\begin{equation}   \label{24}
\left \{
   \begin{array}{lll}
  \hat u_{xx}  \,\,- \frac{\sigma^2}{\varepsilon} \,\,\hat u =\, -\,\frac{1}{\varepsilon} \,\,[\, \,\hat  F(x,s) +u_0(x)\,\,]\\    
\\
  \,\hat u(0,s)\,=\, \hat g_1\,(s)\qquad \hat u(L,s)\,=\,\hat g_2\,(s).
   \end{array}
  \right.
\end{equation}

\noindent   If one introduces the following {\it theta function }

\begin{equation}\,  \label{25}
\hat \theta \,(\,y,\sigma)\,= 
\\  \frac{1}{2 \,\, \sqrt\varepsilon \,\,\,\sigma  } \, \biggl\{\, e^{- \frac{y}{\sqrt \varepsilon} \,\,\sigma}+\, \sum_{n=1}^\infty \,\, \biggl[ \,e^{- \frac{2nL+y}{\sqrt \varepsilon} \,\,\sigma} \, +\, e^{- \frac{2nL-y}{\sqrt \varepsilon} \,\,\sigma}\,
\biggr] \, \biggr\}
\end{equation}
\[
 =\frac{\cosh\,[\, \sigma/\sqrt{\varepsilon} \,\,(L-y)\,]}{\,2\, \, \sqrt{\varepsilon} \,\, \sigma\,\,\, \sinh\, (\,\sigma/\sqrt{\varepsilon}\,\, \,L\,)}\]

\noindent then, by  (\ref{24}) and (\ref{25}) one deduces:

\begin{equation}     \label{26}
\hat u (x,s) = \,\int _0^L \, [\,\hat \theta\,(\,x+\xi, \,s\,)\,-\,\,\,\hat \theta\,(\,|x-\xi|,\, s\,)\,] \, \,[\,u_0(\,\xi\,) \,+\,\hat F(\,\xi,s)\,]\,d\xi\, \end{equation}
\[
 -\, 2 \,\,\varepsilon \, \,\hat g_1 \,(s) \,\,    \hat\theta_x (x,s)\,+ \, 2 \,\, \varepsilon  \,\, \hat g_2 \, (s)\,\,\,\hat  \theta_x \,(L-x,s\,),
\]

\noindent where 
\begin{equation} \label{27}
\hat \theta_x\,\,(x,\sigma)\,=\,\,\frac{\sinh\,[\, \sigma/\sqrt{\varepsilon} \,\,(x-L)\,]}{\,2\, \, \varepsilon \,\, \, \sinh\, (\,\sigma/\sqrt{\varepsilon}\,\, \,L\,)}
\end{equation}

\section{Explicit solution}

In order to obtain the inverse formula for (\ref{26}),         let us  apply (\ref{23}) to (\ref{25}). Then,  one deduces the  following function which is  similar to  {\em theta functions}:

\begin{equation}     \label{31}
\theta (x,t) \,=\,  K(x,t) \ +\, \sum_{n=1}^\infty \,\, \ [\, K(x \,+2nL,\,t) \, + \, K ( x-2nL, \,t)\,] \, =  
\end{equation}
\[\, =\sum_{n=-\infty }^\infty \,\, \ K(x \,+2nL,\,t). \,\]

 \noindent So that, denoting by 
  
\[G(x,\xi, t) \, = \,  \theta \,(\,|x-\xi|,\, t\,)\,- \,  \theta \,(\,x+\xi,\,t\,),
 \] 
 
\noindent   when $\, F\, =\, f(x,t), \,$ by (\ref{26}) the explicit solution of  the {\em linear} problem (\ref{21}) is given by:

\begin{equation}   \label{32}
  u(\, x,\,t\,)\, = \,\,\int^L_0 \,
G(x,\xi, t) \, 
  \,u_0(\xi)\,\, d\xi \,\, \,+\, \,\int^t_0 d\tau\int^L_0 \, G(x,\xi, t) \, \,\,\, f\,(\,\xi,\tau\,)\, \,\,d\xi   \end{equation}

  \[
- \,2 \, \varepsilon \,\int^t_0 \theta_x\, (x,\, t-\tau) \,\,\, g_1 (\tau )\,\,d\tau\,+\, 2\,\, \varepsilon \int^t_0 \theta_x\, (x-L,\, t-\tau) \,\,\, g_2 (\tau )\,\,d\tau\,\]

 So, owing to the basic properties of $ K(x,t), $ it is easy to deduce the following theorem:

  \begin{theorem}
When the linear source $ \, f(x,t)\,  $   and the initial boundary  data $ u_0(x),$ $ g_i\,\,(i=1,2)$  are  continuous in $ \Omega_T,\, $  then  problem $(\ref{21}) $ admits a unique regular solution $ u(x,t)  $ given by  (\ref{32}).
\end{theorem}

Furthermore, when   the source term $\, F\, =\, F(x,t,u) \,$   depends on the unknown function $ u(x,t), $ too, then problem  (\ref{21})     admits  the following  integral equation:

\begin{equation}   \label{33}
 u( x,t)\, = \int^L_0  G(x,\xi, t) u_0(\xi) d\xi +\int^t_0 d\tau\int^L_0  G(x,\xi, t)   F(\xi,\tau,\,u(x,\tau))d\xi- \end{equation}
 
\[\,2 \, \varepsilon \,\int^t_0 \theta_x\, (x,\, t-\tau) \,\,\, g_1 (\tau )\,\,d\tau\,+\, 2\,\, \varepsilon \int^t_0 \theta_x\, (x-L,\, t-\tau) \,\,\, g_2 (\tau )\,\,d\tau\,\]

Indeed,if 
\[\,  F= \, F (x,\,t, \,u(\,x,t\,), \,\, u_x(\ x,t\,)\,\, ),\]
\noindent   let

\begin {center}
$\ D \,= \{ \, (\,x,t,u): ( \, x,t\,) \in \Omega _T, \,\,\,\,\,-\infty <u<\infty  \,\,-\infty <p<\infty.  \, \}, $
\end{center}

\noindent and let us assume   the following  {\bf Hypotheses A}:

\noindent - The function $\,  F (x,t,u,p)\,\,$ is defined and continuous on $\, D  \,$ and  it is bounded for all $ \,u\, $ and  $\, p .\,$

 \noindent -  For each $  \,k\,>\,0\,$ and for $\, |u|,\,|p| <\,k,\,$  the function  $\,  F \,\,$ is Lipschitz continuous in    $ \, \,x\,  $ and  $ \,t\,$ for each compact subset of $\, D.\,$

\noindent - There exists a constant $ \, \beta _F \, $ such that:
\begin {center}
 $\, | F (x,t,u_1,p_1)\,-\, F (x,t,u_2,p_2)|\, \leq \,\beta _{ F} \,\, \{\, | u_1-u_2\,| \, +\,| p_1-p_2\,| \, \} \ $ 
\end{center}

\noindent holds for all  $\,( u_i,\,p_i) \,\, i=1,2. \,$

Moreover,  let us  assume  $\,\, ||\,z\,|| \,= \displaystyle \sup _{ \Omega_T\,}\, | \, z \,(\,x,\,t) \,|, \,\, $ and let $ \,{\cal B}_ T \,  $ denote the Banach space 

\begin {equation}   \label{34}
  \,{\cal B}_ T \, \equiv \, \{\, z\,(\,x,t\,) : \, z\, \in  C \,(\Omega_T),  \, \,\,   ||z|| \, < \infty \ \}.
\end{equation}

  By means of standard methods related to integral equations  and owing to properties of $ K \,\, \mbox{and} \,\, F,\,\, $it is easy  to prove that the  mapping defined by (\ref{33}) is a contraction  of $ {\cal B}_ T $ in $ {\cal B}_ T  $ and so it admits an unique fixed point   $ u(x,t)  \, \in {\cal B}_ T $  \cite{c,dmm}.  Hence

\begin{theorem}

When the  data $ (u_0, g_1,g_2 )$ are continuous functions, then the Dirichlet problem related to the non linear system (\ref{21}), has a unique solution in the space of  solutions which are regular in $ \Omega_T $. 
\end{theorem}

\section{Equivalence with  (ESJJ)}

 Among many others equivalences, a  significant one  occurs between  operator   $\mathcal L$ and equation (\ref{13})  which describes  the evolution inside an exponentially shaped Josephson junction. 

Indeed, let us assume

  \begin{equation} \label{42}
\beta \, = \frac{1}{\varepsilon}\qquad b= \, \beta^2 \, (1-\alpha \, \varepsilon \,) \qquad  a \, \beta  = \,\, \frac{\lambda^2}{4} \,-\,b \end{equation}
\[ 
   f\,= \, - \int_0^t\, \, e^{-\, \frac{1}{\varepsilon}(t-\tau) } f_1 (x,\tau, u) \,d \tau, \,\]

\noindent with 

\begin{equation}  \label{41}
 f_1 =\, e^{-\frac{\lambda} {2}\,x\,} [\, \sin \,( e^{\,x\,\lambda /2\,}\,   u ) \, - \gamma]. 
  \end{equation}
  
\noindent From the integro differential equation (\ref{11}) follows:

\begin{equation}  \label{43}
\varepsilon  u_{xxt}\, - \,  u_{tt} \, +\,  u_{xx}-\, (\alpha \,+ \varepsilon \, \frac{\lambda^2} {4})  u_t\,\,-\,  \frac{\lambda^2} {4}\, u = \, f_1   
\end{equation}

\noindent

\noindent Therefore, assuming   $ e^{\frac{\lambda} {2}\,x\,} \,  u =\bar u,\,   $    (\ref{43}) turns into equation  (\ref{13}).

 The diffusion effects due to the  dissipative terms $\, \varepsilon \,( u_{xxt}\,- \lambda u_{xt})$ of (\ref{13}) have  already been investigated in \cite{df213}. Here, in sect \ref{sec:asympotic}, by means of this equivalence, the influences of  data on the solution which is related to the Dirichlet problem, will be explicitly evaluated . 

\section{Some properties on $ \theta  $ and $ \theta_x $ functions }

In order to obtain a priori estimates and asymptotic effects related to boundary perturbations, some properties  of the fundamental solution $ K  $ and the theta function defined in (\ref{31}) ought to be   evaluated. In \cite{mda10,de13,dr13,dr8} some of these have  already been proved and it results:

\begin{equation}               \label{51}
 \int_0^L |\theta (|x-\xi|,\,t)|\ \, d\xi \leq \,  ( 1\, +\, \sqrt b \,\pi \,t \, ) \,\,e^{- \omega \, t\,} \qquad \omega= min\{a, \beta\}
\end{equation}

 \begin{equation}               \label{52}
\biggr| \,  \int_0^L\,\theta (|x-\xi|,\,t)\,\,d\xi\,\biggl|\,\,\, \leq \,\sum_{n= -\infty }^\infty \, \, \int_0^L\,| K(|x-\xi +2nL|, \,t)| \,\,d\xi\,\ =
\end{equation}
 
\[ =\,\sum_{n= -\infty }^\infty \, \, \int_{x+(2n-1)L}^{x+2nL}\,| K(y,\,t)| \,dy\,\,\,\,\leq  \,\, \,\int_\Re\,\,|K(y,t)|\,d y .\,\,\]

 \noindent \noindent  Moreover, letting: 

  \begin{equation}     \label {53}
 K_{1}\,\equiv\,   \int^t_0 \,\,e^{-\,\beta \,(\,t-\tau)\,}\,K\,(x,\tau\,) \, d \tau\,,     
 \end{equation} 
\noindent one has:

 \begin{equation}   \label{54}
\int_\Re |K_1| \, \ d\xi \leq \, \,E(t);\,  \qquad  \int _0^t\\d\tau \,\int_\Re |K_1| \, \ d\xi \leq \, \beta_1\,
\end{equation}

  \noindent being \[ E(t) \,=\, \frac{e^{\,-\,\beta t}\,-\,e^{\,-\,at}}{a\,-\,\beta}\,\,>0\, \qquad \mbox{and }\qquad \beta_1\,=\, ({a\,\beta})^{\,-1}.\,\]
 
 \noindent Furthermore, as for $  \theta_x $, from (\ref{22}), it is well-rendered  that the x derivative of the  integral term vanishes  for $ x\,\rightarrow 0 \,$, while the first term represents  the  derivative with respect to  $ x $ of  the fundamental solution related to the heat equation. Indeed, since  (\ref{31}), it results:

 \begin{equation}
  \theta_r \, = K_r \, + \sum_{n=1}^\infty \,\, \ [\,-\frac{x+2nl}{2t}\,\, K(x \,+2nL,\,t) \, - \, \frac{x-2nl}{2t} \,\,K ( x-2nL, \,t)\,] \, \end{equation}
  
 \[ =\, K_r +J(r,t)\]

 \noindent and, by means of classic calculations,(see, f.i. \cite{c}), one has:  
 
 \begin{equation}
 \lim _{x \to  0} \, J(r,t) \,=\,0.
 \end{equation}
 
 \noindent So,   classic theorems assure that  conditions $(\ref{21})_3$ are surely  satisfied.

 \noindent Moreover, being:

\begin{equation} 
\lim _{t \to \infty} \int _0^t \theta_x ( x, \tau ) \,\, d\tau \,\,= \lim _{s \to 0} \hat \theta_x(x,s) \qquad  Re \,\,s > max \{-a,-\beta\}
\end{equation}

\noindent  from   (\ref{27}), it results:
\begin{equation}              \label{56}
\lim _{t \to \infty} \int _0^t \theta_x ( x, \tau ) \,\, d\tau \,\, = \frac{1}{2 \, \varepsilon \,  }\,\,\ \frac{\sinh \sigma_0 \,\,(x-L)}{\sinh\,(  \sigma_0 \, L) } \end{equation}

\noindent  where $ \sigma_0 = \sqrt{\biggl(\,a\,\,+ \frac{b}{\beta}\biggr)\frac{1}{\varepsilon}}.$

Besides,  for all $ x>\delta>0 $ it results:

\begin{equation}
\sum_{n=-\infty }^\infty \,\,e^{ -\,\,\frac{(x+2nL)^2}{4 \varepsilon \, t}\,} \, \leq \,\,t\, \,\frac{4\varepsilon \pi^2}{4 \,l^2 e} \, \csc^2 (\pi \delta  /2L), 
\end{equation}

\noindent  and  one has  
\begin{equation}
| K_r ( r,t )| \leq \,r\,\,\frac{e^{- \frac{r^2}{4 t}\,}}{4\,\sqrt{\pi \varepsilon t^3}} \,\, [ \, 1\, +\, 4\,b \,t^2 \,]\,e^{-\,\omega \, t}.\,  
\end{equation} 
\noindent So, the  following theorem can be proved:

\begin{theorem}
For all $ 0\leq x  \leq L,$ there exists  a positive constant  $C $ depending only on constants   $a, \varepsilon ,  b,\, \beta $  such that:  
 
\begin{equation} \label{511}
 \int_0^ \infty \, | \theta_x ( x,\tau )| \,\, d \tau  \,\, \leq \,\,C.
 \end{equation}
 Moreover, it results:

\begin{equation}              \label{512}
\lim _{t \to \infty}  \theta_x ( x, t ) \,\, = \,\,0,
\end{equation}

\end{theorem}

\section{Analysis of boundary contributions}

By means of the following well known theorem (see f.i. \cite{dr13} and references therein),    boundary contributions related to  the Dirichlet problem  (\ref{21}) can be explicitly  evaluated.

\begin {theorem}
 Let  $ h(t) $ and $ \chi(t) $  be two continuous functions on $ [0,\infty [$      satisfying  the following  hypotheses 

 \noindent 
\begin{equation}  \label{hp}
\exists \,\, \displaystyle{\lim_{t \to \infty}}h(t) \, = \, h(\infty),\qquad\exists \,\, \displaystyle{\lim_{t \to \infty}}\chi(t) \, = \, \chi(\infty)\,
\end{equation}  
 
 \noindent 
\begin{equation} \label{hp2}
 \dot  h(t)\,  \in \, L_1  [ \,0, \infty).\end{equation}

 \noindent Then, it results: 

\noindent 
 \begin{equation}      \label{63}
\lim_{t \to \infty} \,\, \int_o^t \,\chi(t-\tau ) \, \dot h ( \tau ) \, d \tau \,\, = \, \,\chi(\infty) \,\, [\,\,h(\infty) - h(0)\,\,].
\end{equation} 
\end{theorem} 

\noindent According to this, the following theorem holds:

\begin{theorem} \label{theorem asintotico}
If   data $ g_ i  \,\,\ (i=1,2) \,\,$  are two  continuous functions  satisfying condition (\ref{hp}) then, one has:

\begin{equation}    \label{64}
\lim_{t \to \infty } \,\int_0^t \,\theta_x \,(x,\tau)\,\,\, g_i \,(t-\tau)\, \,d\,\tau \, = \, g_{i, \infty} \,\,\,\,\,  \frac{1}{2 \, \varepsilon \,\,\,  }\,\,\ \frac{\sinh \sigma_0  \,\,(x-L)}{\sinh\  \sigma_0  \, L }
\end{equation}

\noindent being  $ g_{i,\infty }= \lim_{t \to \infty } g_i , \,\, (i=1,2)$  and $ \sigma_0 = \sqrt{\biggl(\,a\,\,+ \frac{b}{\beta}\biggr) \frac{1}{\varepsilon}}.$

Moreover, when  data  $\varphi _ i \,\,\, ( i=1,2) $ verify both condition  (\ref{hp}) and condition (\ref{hp2}),  then it results: 

\begin{equation}    \label{65}
\lim_{t \to \infty } \,\int_0^t \,\theta_x \,(x,\tau)\,\,\, \varphi_i \,(t-\tau)\, \,d\,\tau \, = \,0\,  \qquad ( i= 1, 2) 
\end{equation}
\end{theorem}

\begin{proof}
Let us assume,  in (\ref{63}), $ h= \int _0^t \theta_x(x,\tau) d\tau \, \,\,\mbox{and}\,\, \chi = g_i  \,\,(i=1,2) $. Then,  (\ref{64}) follows  by (\ref{56}) and (\ref{511}).

Besides, when  $ \chi= \theta_x(x,t)  $ and $ h= g_i, $ since   (\ref{512}), (\ref{65}) is deduced.

\end{proof}

\section{ Asymptotic behaviours  related to (ESJJ)}  \label{sec:asympotic}

According to the equivalence between  (\ref{13})  and the integro differential equation (\ref{11}), previous results can be applied to (ESJJ). 

Let us consider the following initial boundary  value problem with Dirichlet conditions:  

 \begin{equation}            \label{71}
 \left \{
   \begin{array}{lll}
\varepsilon u_{xxt}+u_{xx}-u_{tt} - \varepsilon\lambda u_{xt} -\lambda  u_x- \alpha u_{t}=\sin
  u-\gamma& (x,t) \in \Omega_T \,  \\    
\\  \,u (x,0)\, = u_0(x)\, \,\qquad u_t(x,0) = v_0(x)\,\, &
x\, \in [0,L], 
\\
\\
  \, u(0,t)\,=\,g_1(t)  \,\,\,\qquad u(L,t)\,=\,g_2(t) & 0<t\leq T.
   \end{array}
  \right.
 \end{equation}

 When in (\ref{21}) one assumes 
 \begin{equation}    \label{72}
F(x,t,u) =  e^{-\frac{\lambda} {2}\,x\,} \,\,\biggl [ \int_0^t\, \, e^{-\, \frac{1}{\varepsilon}(t-\tau) } [\, \sin \,( e^{\,x\,\lambda /2\,}\,   u ) \, - \gamma] \,d \tau \,\,-\, \, v_0(x) \,\, e^{-\frac{t}{\varepsilon }}\biggr].
 \end{equation} 
\noindent  and $ a, b, \beta , \lambda $ are given by  $(\ref{42})_1$, system (\ref{71}) can be given the form (\ref{21}) and  hypotheses A assure that the following integro equation states:

\begin{equation}   \label{73}
 u( x,t)\, = \int^L_0  G(x,\xi, t) \,\, e^{-\frac{\lambda} {2}\,x\,} \, u_0(\,\xi) \,d\xi  \end{equation}

\[+\int^t_0 d\tau\int^L_0  G(x,\xi, t) \,\, F(\xi,\tau,\,u(x,\tau))\,\,d\xi\] 
\[-\,2 \, \varepsilon \,\int^t_0 \theta_x\, (x,\, t-\tau) \,\,\, g_1 (\tau )\,\,d\tau\,+\, 2\,\, \varepsilon \int^t_0 \theta_x\, (x-L,\, t-\tau) \,\,\, g_2 (\tau )\,\,d\tau\,.\]
 
So, a priori estimates of the solution and   asymptotic effects of boundary perturbations can be obtained. Indeed, letting

\[||\,u_0\,|| \,= \displaystyle \sup _{ 0\leq\,x\,\leq \,L\,}\, | \,u_0 \,(\,x\,) \,|, \,\,\qquad ||\,v_0\,|| \,= \displaystyle \sup _{ 0\leq\,x\,\leq \,L\,}\, | \,v_0 \,(\,x\,) \,|,  \, \]

\noindent and 
\[  
||\,f|| \,= \displaystyle \sup _{ \Omega_T\,}\, | \,f \,(\,x,\,t,\,u) \,|,\]

\noindent with  $ f$  defined by $(\ref{42})_2$-(\ref{41}),  the following theorems hold:

\begin{theorem}
When  $g_1\,\,= g_2 \,= \,0,\, \, $ the  solution  $ u(x,t)  $  related to (ESJJ) model,  and determined in (\ref{33}),  satisfies the following estimate:

\begin{equation}            \label{55}                      
 \left| u \, \right| \, \leq  2\,[\,\left\| u_0 \right\| \, (1+\pi \sqrt b \, t ) \, e^ {\,-\omega\,t\,}\,+\,\left\| v_0 \right\|\,E(t)+ \,  \beta_1\, \left\| f \right\| \,] 
    \end{equation}
    
    \begin{proof}
The initial disturbances can be evaluated by means of (\ref{51}),(\ref{52}) and $(\ref{54})_1$. Besides, since  $(\ref{54})_2$, the behaviour of the source term is determined.     
    \end{proof}
 \end{theorem}
Further, according to theorem (\ref{theorem asintotico}),  if $ u_0 = v_0=0 $ and $ F=0, $ one has:  

\begin{theorem}

  When   $ t $ tends to infinity and  data $ g_ i  \,\,\ (i=1,2) \,\,$  are two  continuous functions  satisfying condition (\ref{hp}),  it results:  

\begin{equation}             \label{75}                                                 
 u  \, = g_{1,\infty}\,\,\frac{\sinh \frac{\lambda}{2}  \,\,(x-L)}{\,2 \varepsilon \,\,\sinh\  \frac{\lambda \,L}{2} \,}  + g_{2,\infty}\,\,\frac{\sinh [\frac{\lambda\, }{2} \,(x-2L)  \,] }{\,2 \varepsilon \,\,\sinh\ \frac{\lambda\,L}{2}  }
 \end{equation}
\noindent   Otherwise, in the hypotheses  (\ref{hp}) (\ref{hp2}),     
the effects determined by boundary disturbance  vanish. 
\end{theorem}

 {\bf Remarks}. When   $ t $ tends to infinity,  the effect due to the initial disturbances $\, (\,u_0, v_0\,) \, $ is  vanishing,  while  the  effect of the non linear source is bounded for all $ t. $ Furthermore, for large t, the  effects due to  boundary disturbances  $ g_1, g_2 $     are null or at least    everywhere bounded.

\begin{acknowledgements}
 This paper has been performed under the auspices of
G.N.F.M. of I.N.D.A.M. \end{acknowledgements}



\end{document}